\documentclass[aoas,preprint]{imsart}
\setattribute{journal}{name}{} 

\RequirePackage[OT1]{fontenc}
\RequirePackage{amssymb,amsthm,amsmath}
\RequirePackage{natbib}
\RequirePackage[colorlinks,citecolor=blue,urlcolor=blue]{hyperref}
\usepackage{latexsym}
\usepackage{graphicx}
\usepackage{subfigure}
\usepackage{wrapfig,epsfig}
\usepackage{caption}
\usepackage{booktabs}
\usepackage{float}
\usepackage{ctable,multirow}
\usepackage{color}

\startlocaldefs
\numberwithin{equation}{section}
\theoremstyle{plain}

\endlocaldefs

 \newcommand{\calJ}{\mathcal{J}}
 \newcommand{\bbR}{\mathbb{R}}
 \newcommand{\bbS}{\mathbb{S}}

 \newcommand{\argmin}{\mathop{\rm arg\,min}}

 \newcommand{\trace}{\mathop{\rm trace}}
 
 \newcommand{\Range}{\mathop{\rm Range}}

 \newcommand{\var}{\mathop{\rm var}}
 \newcommand{\E}{\mathop{\rm E}}
 
 \newcommand{\diag}{\mathop{\rm diag}}

 \newcommand{\clr}{\mathop{\rm clr}}

\newtheorem{theorem}{Theorem}[section]
\newtheorem{proposition}[theorem]{Proposition}

\newtheorem{definition}[theorem]{Definition}

\begin{document}

\begin{frontmatter}
\title{Kernel-Penalized Regression for Analysis of Microbiome Data}
\runtitle{Kernel-Penalized Regression}

\begin{aug}
\author{\fnms{Timothy W.} \snm{Randolph}\thanksref{m1}\ead[label=e1]{trandolp@fredhutch.org}},
\author{\fnms{Sen} \snm{Zhao}\thanksref{m2}},
\author{\fnms{Wade} \snm{Copeland}\thanksref{m1}},
\author{\fnms{Meredith} \snm{Hullar}\thanksref{m1}}
\and
\author{\fnms{Ali} \snm{Shojaie}\thanksref{m2}
}


\affiliation{Fred Hutchinson Cancer Research Center\thanksmark{m1} and University of Washington\thanksmark{m2}}


\end{aug}

\begin{abstract}
The analysis of human microbiome data is often based on dimension-reduced graphical
displays and clusterings derived from vectors of microbial abundances in each sample.
Common to these ordination methods is the use of biologically motivated definitions of
similarity.  Principal coordinate analysis, in particular, is often performed using
ecologically defined distances, allowing analyses to incorporate context-dependent,
non-Euclidean structure. In this paper, we go beyond dimension-reduced ordination methods and
describe a framework of high-dimensional regression models that extends these
distance-based methods.  In particular, we use kernel-based methods to show how to
incorporate a variety of extrinsic information, such as phylogeny, into penalized
regression models that estimate taxon-specific associations with a phenotype or
clinical outcome. Further, we show how this regression framework can be used to address
the compositional nature of multivariate predictors comprised of relative abundances;
that is, vectors whose entries sum to a constant. We illustrate this approach with
several simulations using data from two recent studies on gut and vaginal microbiomes.
We conclude with an application to our own data, where we also incorporate a
significance test for the estimated coefficients that represent associations between
microbial abundance and a percent fat.
\end{abstract}



\end{frontmatter}

\section{Introduction}

A common tool in the analysis of data from microbiome studies is a scatterplot of
dimension-reduced microbial abundance vectors.  This is a display of the samples' beta
diversity which, in ecology, refers to differences among various habitats. When applied
to human studies, beta diversity describes the variation in microbial community structure
across sampling units (e.g., human subjects): a beta diversity plot displays the $n$
sampling units with respect to the principal coordinates of their microbial abundance
vectors, each consisting of measures on the $p$ taxa observed in the study; see, e.g.,
\citet{Claesson:2012,Koren:2013,Kuczynski:2010,Goodrich:2014a}. This principal coordinates
analysis (PCoA; or multidimensional scaling, MDS) begins with an $n\times n$ matrix of
pairwise dissimilarities between abundance vectors. The choice of dissimilarity measure
may greatly influence the biological interpretation \citep{Lozupone:2007,Fukuyama:2012}.
Euclidean distance is rarely used.

Dissimilarity measures that account for phylogenetic relationships among the taxa are
assumed to enhance statistical analyses --- for instance, to improve the power of
statistical tests --- because they incorporate the degree of divergence between sequences
\citep{Chen:2012} and do not ignore ``the correlation between evolutionary and ecological
similarity" \citep{Hamady:2009}. The UniFrac distance \citep{Lozupone:2005}, in
particular, is based on the premise that taxa which share a large fraction of the
phylogenetic tree should be viewed as more similar than those sharing a small fraction of
the tree. In the unweighted version of UniFrac, each taxon is quantified merely by its
presence or absence; the distance between a pair of samples is based on the number of
branches in the tree shared by both. Figure~\ref{fig:PCoA-UniFrac-age}(a) is a beta
diversity plot of $n=100$ human microbial abundance vectors with $p=149$ taxa based on
data from \citet{Yatsunenko:2012}. The 2-dimensional coordinates of the samples are
displayed with respect to the unweighted UniFrac distance, and each sample is colored
according to the age of the subject.

\begin{figure}
\centering
{\includegraphics[width=\textwidth]{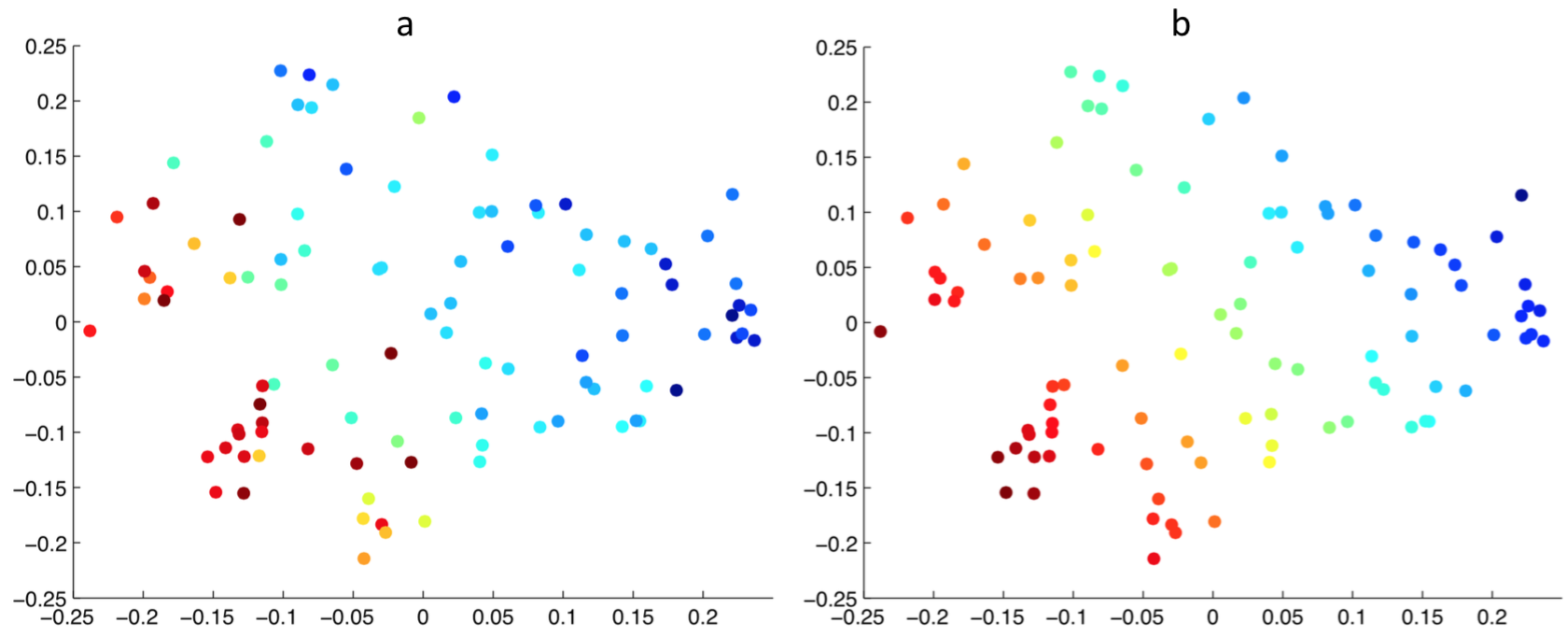}}
\caption{\small \textbf{PCoA plots of data from \citet{Yatsunenko:2012}.} \textbf{(a)}: PCoA plot with respect to unweighted UniFrac distance, colored according to log(age) of subject.
\textbf{(b)}: PCoA plot with respect to unweighted UniFrac distance, colored according to
$y_{\mbox \footnotesize True}$ from the model in Eq.~\eqref{eq:KPCR-model} with $\epsilon=0$.}
\label{fig:PCoA-UniFrac-age}
\end{figure}

Dissimilarity measures in microbiome studies are many and varied, with a rich collection
that, like UniFrac, exploit the phylogentic structure: \citet{Chen:2012} generalize
UniFrac by reweighting rare and abundant lineages; double principal coordinate analysis
(DPCoA) \citep{Pavoine:2004}, as shown by \citet{Purdom:2011}, generalizes PCA by
incorporating the covariance that would arise if the data was created by a process
modeled by the tree; the \emph{edge PCA} method of \citet{Matsen:2013} incorporates taxon
abundance information at all nodes in a phylogenetic tree, rather than just the leaves of
the tree, and \citet{EvansMatsen:2012} formalize the mathematical interpretation of
UniFrac as just one example within a large family of Wasserstein (or earth mover's)
metrics.  A wide variety of non-phylogenetic dissimilarities are also in common use, such
as Bray-Curtis \citep{HMP:2012} and Jenson-Shannon \citep{Koren:2013}, among others.


While PCoA plots provide valuable graphical insight into the relationships among
microbial profiles and an outcome or phenotype, they do not \emph{quantify} this
association. More importantly, the (sets of) taxa associated with the outcome --- and the
magnitude or statistical significance of such associations --- are not ascertained from a
PCoA plot; once a matrix of (dis)similarities between samples is formed, it is not clear
how to identify individual taxa that are associated with an outcome. Specifically, given
a PCoA plot as in Figure~\ref{fig:PCoA-UniFrac-age}(a), with structure imposed by the
chosen dissimilarity matrix (e.g., unweighted UniFrac) and with associations implied by a
class label or continuous outcome (e.g., age), how does one estimate which taxa or
subcommunities are associated with this outcome?  We address this question by formulating
multivariate regression models that are constrained by the structure of the
(dis)similarity matrix. This is made possible by exploiting an equivalence between a
taxon-based (primal space) and sample-based (dual space) formulation of our penalized
regression models. While exploiting such an equivalence is straightforward in the special
case of ridge regression (with purely Euclidean structure), it becomes complicated when
more general distance measures are used. To this end, we show how a little-used
regularization scheme by \citet{Franklin:1978} provides a dual-space regression
coefficient estimate that naturally connects to primal-space coefficients. Because a
dissimilarity matrix can be used to construct a similarity matrix (as commonly done in
classical MDS \citep{MardiaKentBibby:1980}), we work with kernels, rather than distances,
and allow for general kernels, including those constructed from a nonlinear feature map.

In addition to complications stemming from more general distances, the analysis of
microbiome data is also complicated by the compositional nature of the data itself. More
specifically, taxon measures typically represent \emph{relative}, rather than absolute,
abundances. The $p$-variate relative abundance vectors are thus \emph{compositional} in
that they are constrained to a simplex within $\bbR^p$; such data do not reside in a
Euclidean vector space \citep{Aitchison:2003}. Consequently, spurious correlations arise
and standard multiple regression models fail.
Our proposed KPR framework, however, addresses this: the centered log (CLR) transform of
the relative abundance vectors first removes the vectors from the simplex, then the
estimation process is constrained using a penalization term defined by Aitchison's
variation matrix. This approach takes a different perspective from the recent proposal of
\citet{Li:2014} which forces the estimated coefficient vector to reside in the simplex.
Given that the CLR transforms the compositional vectors to Euclidean space and that the
units of the Aitchison variation matrix are the same as the CLR transformed data
\citep{Egozcue-Chap2:2011}, our constraint seems more suitable for the geometry of the
problem.

In summary, we describe a family of high-dimensional regression problems in
Section~\ref{sec:PenReg}, which are designed to incorporate the assumptions that are
tacitly implied by various exploratory and graphically-focused PCoA plots common in
microbiome studies. We show how phylogenetic and other structure can be incorporated via
kernel penalized regression in either the primal ($p$-dimensional) feature space or the
dual ($n$-dimensional) samples space; see Sections~\ref{sec:DPCoA-PenReg} and
\ref{sec:KPR}. Finally, our proposed framework leads to an approach, described in
Section~\ref{sec:compdata}, for addressing well-known problems that arise from applying
standard (Euclidean-based) statistical models to compositional data.
Section~\ref{sec:Examples} illustrates the proposed framework with simulations based on
publicly available data, while Section~\ref{sec:Applic} presents an application to our
recent microbiome study of premenopausal women. In this analysis, we obtain estimates of
associations between microbial species and percent fat measured in premenopausal women,
and also provide inference for these estimates by applying a recent significance test
\citep{ZhaoShojaie:2015} in our kernel-penalized regression (KPR) framework.

\section{Kernel Penalized Regression for Microbiome Data}\label{sec:PenReg}

We describe a family of multiple regression problems aimed at incorporating assumptions
that are implicit in PCoA plots common in microbiome studies.
We begin in Section~\ref{sec:PCoA-PCR} by establishing notation and concepts from
existing dimension-reduction (ordination) methods with the goal of extending them to
non-truncated (penalized) regression models.
Section~\ref{sec:DPCoA-PenReg} extends PCoA and PCR to penalized regression models in the
primal space in a manner that incorporate structures implicit in recent microbiome
analyses.
Section~\ref{sec:KPR} extends kernel ridge regression to general (non-$L^2$) structure
and the use of two kernels. This extension exploits a \emph{dual-space} regularization
scheme of Franklin \citep{Franklin:1978}.
Section~\ref{sec:compdata} describes how our proposed framework can be applied to
formulate a penalized regression model that accounts for the structure of compositional
data.

We denote by $y_i$, $i=1,.., n$, a real-valued quantified trait, and by $x_i =
[x_{i1},..., x_{ip}]'$ a $p$-dimensional vector of microbial abundance values measured
for each of $n$ subjects. Denote by $X$ the $n\times p$ sample-by-taxon matrix whose
$i$th row is $x_i'$. We assume throughout that the columns of $X$ are mean centered. For
now, we assume that the abundance values are appropriately normalized/transformed and
postpone the treatment of compositional data to Section~\ref{sec:compdata}.  The
transpose of a matrix $A$ is denoted by $A'$ and the Frobenius norm is denoted as
$\|A\|_F$.  The Euclidean norm of a vector $x\in\bbR^p$ is denoted $\|x\|_{\bbR^p}$,
$\|x\|_{2}$ or simply $\|x\|$.

\subsection{Background for PCoA and principal component regression}\label{sec:PCoA-PCR}

Consider first the Euclidean PCoA, which is obtained from the eigenvectors of the
\emph{kernel} matrix $K_I:=XX'$ of inner products $K_{ij}=\langle x_i,x_j\rangle$ between
samples. Let $\calJ$ be the centering matrix, $\calJ = I -\frac{1}{n}
\boldsymbol{1}\boldsymbol{1}'$, where $\boldsymbol{1}$ is the $n\times 1$ vector of ones.
Then, it can be seen that $XX' = -\frac{1}{2}\calJ\Delta^E\calJ$, where $\Delta^E$ is the
$n\times n$ matrix of {\em squared} Euclidean distances between samples:
$\Delta^E_{i,j}=\|x_i-x_j\|_{\bbR^p}^2$. The relationship between a kernel and a distance
matrix $\Delta$ is more general. In particular, if $\Delta$ is any $n\times n$ symmetric
matrix of squared dissimilarities between vectors in $\bbR^p$ then $H
=-\frac{1}{2}\calJ\Delta\calJ$ serves as a kernel matrix summarizing similarities; see,
e.g., \citet{Gower:1966,Pekalska:2002}. A particular case involves a $p\times p$
symmetric, positive definite matrix $Q$ that defines an inner product $\langle
x_i,x_j\rangle_Q=x_i'Qx_j$ on $\bbR^p$. If $\Delta^Q$ denotes the matrix of squared
distances, $\Delta^Q_{i,j}=\|x_i-x_j\|_{Q}^2 = \langle x_i-x_j, x_i-x_j\rangle_Q$,
defined with respect to this inner product, then $XQX'=-\frac{1}{2}\calJ\Delta^Q\calJ$ is
also a similarity kernel for the $n$ samples. We will denote this kernel by $K_Q=XQX'$.
Similarly, one may start with a matrix $\Delta^U$ of squared distances defined by a
tree-based UniFrac dissimilarity \citep{Lozupone:2005}, and define a similarity kernel by
$H=-\frac{1}{2}\calJ\Delta^U\calJ$.

In graphical displays, two or three coordinates are typically used to explore the
relationship between samples. Let $K = US^2U'$ be the eigen-decomposition of any
similarity kernel, $K$, where $U$ is the matrix whose columns are eigenvectors and
$S^2=\diag\{\sigma_j^2\}$ is the diagonal matrix of eigenvalues. The two-dimensional PCoA
plot is then the collection of points $\{\eta_{i1}, \eta_{i2}\}_{i=1}^n :=\{(\sigma_1
U_{i1}, \sigma_2 U_{i2})\}_{i=1}^n$; i.e., a plot of the points represented by the first
two columns of the matrix $US$. These points are often colored according to a grouping
label or continuous value, $\{y_i\}_{i=1}^n$, to graphically explore the existence of an
association between the outcome $y$ and the sample profiles summarized by the first few
columns of $US$.  So, a PCoA plot is a graphical depiction of a two-component regression
model of association:
\begin{equation}\label{eq:yreduced}
y_i = \gamma_1\eta_{i1} + \gamma_2\eta_{i2}+\epsilon, \quad i=1,\ldots,n,
\end{equation}
where $\eta_1$ and $\eta_2$ are the first two PCoA axes.  Ordinary principal component
regression (PCR) corresponds to the case that $\eta_{1}$ and $\eta_{2}$ come from the
Euclidean kernel $K_I=XX'$.  On the other hand, the configuration of points in
Figure~\ref{fig:PCoA-UniFrac-age}(b) correspond to the first two eigenvectors of the
kernel defined by an unweighted UniFrac distance matrix $\Delta^U$, and colors of
individual points correspond to the values of $y$ from eq.~\eqref{eq:yreduced} with
$\epsilon = 0$.

Let $A_{(k)}$ denote the first $k$ columns of a matrix $A$, or its first $k$ rows and
columns if $A$ is diagonal. Then, using the singular value decomposition (SVD), $X=USV'$,
if we express the dimension-reduced approximation of $X$ as $\breve{X} :=
U_{(2)}S_{(2)}V_{(2)}'$, then eq.~\eqref{eq:yreduced} can be written as
\begin{equation}\label{eq:gammaPCR}
  \begin{aligned}
     y &= \gamma_1\eta_{1} + \gamma_2\eta_{2}+\epsilon \\
       &= U_{(2)}S_{(2)}\,\gamma + \epsilon \\
       &= \breve{X}V_{(2)}\,\gamma + \epsilon,
  \end{aligned}
\end{equation}
where $\gamma =[\gamma_1 \ \gamma_2]'$. Here, $\breve{X}V_{(2)}=U_{(2)}S_{(2)}$, and
$\Range(\breve{X}')=\Range(V_{(2)})$. Therefore, assuming a coefficient vector $\beta$ of
the form $\beta=\breve{X'}\gamma$, the model $y = \breve{X}V_{(2)}\,\gamma + \epsilon$
can be written as $y=\breve{X}\beta+\epsilon$. So inherent in a Euclidean PCoA plot is an
implicit coefficient vector, $\beta$, which models a linear association between $y$ and
$\breve{X}$. Using the SVD of $X$ in \eqref{eq:gammaPCR}, the PCR estimate of
$\beta\in\bbR^p$ is expressed as
\begin{equation}\label{eq:betaPCR}
  \hat\beta_{\mbox{\tiny PCR}} = (\breve{X}'\breve{X})^{\dag}\breve{X}'y
    = V_{(2)}S_{(2)}^{-1}U_{(2)}'y
    = \sum_{k=1}^2\frac{1}{\sigma_k}u_k'y\,v_k,
\end{equation}
where $\dag$ denotes the Moore-Penrose inverse.


\subsection{Penalized regression and DPCoA}\label{sec:DPCoA-PenReg}
An alternative to a Euclidean PCR is the ordinary ridge regression~\citep{HoerlKen:1970},
\begin{equation}\label{eq:betaRidgeSVE}
  \hat\beta_{\mbox{\tiny ridge}} = (X'X+\lambda I)^{-1}X'y
   = \sum_{k=1}^n\left(\frac{\sigma_k^2}{\sigma_k^2+\lambda^2}\right)\frac{1}{\sigma_k}u_k'y\,v_k,
\end{equation}
in which the terms are re-weighted instead of being truncated, as in
$\hat\beta_{\mbox{\tiny PCR}}$. The estimate in~\eqref{eq:betaRidgeSVE} is the solution
of the penalized least squares regression problem, $\hat\beta_{\mbox{\tiny ridge}} =
\argmin_{\beta}\left\{\|y-X\beta\|^2 + \lambda\|\beta\|^2\right\}$, where here and
throughtout $\lambda$ is a tuning parameter that controls the amount of shrinkage or size
of $\beta$ in the penalty term.  Here the penalty is simply the Euclidean (or $\ell^2$)
norm on $\bbR^p$, but a wide range of penalty terms have been proposed to replace or
extend this particular form of regularization; see \citet{Buhlmann:2014} for a review of
the most established methods. These methods, such as the lasso, elastic net or SCAD do
not incorporate any extrinsic information, but a variety of other penalization methods
have been proposed which aim to do this. For instance, \citet{Tanaseichuk:2014} uses a
tree-guided penalty~\citep{Kim_Xing:2010} to incorporate such structure into a penalized
logistic regression framework to encourage similar coefficients among taxa according to
their relationships in the phylogenetic tree. \citet{TibshTaylor:2011} study the solution
path for computing a ``generalized lasso" estimate in which an $\ell^2$ penalty is
replaced with an $\ell^1$ penalty applied to a linear transformation of the features,
$\lambda\|L\beta\|_1$. Within the context of genetic networks, \citet{LiLi:2008}
accounted for network structure by augmenting the $\ell^1$ penalty with a second penalty
of the form $\lambda_2\|\beta\|_{\mathcal{L}}^2 = \beta'\mathcal{L}\beta$, where
$\mathcal{L}$ denotes the graph Laplacian matrix corresponding to pre-defined connections
between genes in a pathway.

For now, we consider a positive definite $p\times p$ matrix $Q$ with a Cholesky
decomposition $Q=LL'$, and a penalty term of the form $\|L^{-1}\beta\|^2 =
\|\beta\|^2_{Q^{-1}}=\beta'Q^{-1}\beta$.  The generalized ridge (or Tikhonov
regularization \citep{GolubVanLoan:2012}) estimate with respect to $Q$ is then defined as
\begin{equation}\label{eq:hatBetaQ}
\begin{aligned}
 \hat\beta_Q 
        & =\argmin_{\beta}\{ \|y-X\beta\|^2 + \lambda\|\beta\|_{Q^{-1}}^2\}
        = (X'X + \lambda Q^{-1})^{-1}X'y \\
       & = \sum_{k=1}^{n} \left(\frac{\sigma_k^2}{\sigma_k^2+\lambda \mu_k^2}\right) \frac{1}{\sigma_k} u_k'y \, v_k,
\end{aligned}
\end{equation}
This estimate takes the same form as \eqref{eq:betaRidgeSVE} but now the vectors $u_k$
and $v_k$ arise from the SVD of $XL=USV'$.

As an aside, it is worth noting that if $A$ denotes any matrix with $p$ columns, the
structure of an estimate $\hat\beta_A$ from a penalty term of the form $\|A\beta\|_{2}^2$
is determined by the joint eigenstructure of the pair $(X,A)$ via the generalized
singular value decomposition.\footnote{We refer here to the generalized singular value
decomposition (GSVD) of \citet{VanLoan:1976}, a simultaneous diagonlization of two
matrices. A different SVD generalization \citep{Greenacre:1984} imposes constraints on
left and right singular vectors of a matrix.} In particular, the basis expansion of
$\hat\beta_Q$ in \eqref{eq:hatBetaQ} is given in terms of the generalized singular
vectors of $(X,L^{-1})$. Although the ridge estimate (with $Q=I_p$) is biased, an
informed choice of penalty term can, in fact, reduce the bias \citep{Randolph:2012}.

Now consider the context of phylogentic information and let $\delta$ represent the matrix
of squared patristic distances between pairs of taxa --- i.e., the sum of branch lengths
between each pair of taxa on the leaves of a phylogenetic tree.  Set
$Q=-\frac{1}{2}\calJ\delta\calJ$, a matrix of similarities between taxa.  Double
principal coordinate (DPCoA) analysis was proposed as a multi-step procedure by
\citet{Pavoine:2004} to provide an alternative to ordinary PCoA by incorporating both the
structure among samples and the structure implied by species distribution among
subcommunities; i.e., the phylogeny as summarized by $Q$. \citet{Purdom:2011} clarified
the original multi-step DPCoA procedure and showed how it can be more simply understood
as a generalized PCA (gPCA) in which one obtains the new coordinates from the
eigenvectors of $K_Q=XQX'$. 
Note that when $Q$ is the identity matrix ($Q = I_p$), DPCoA reduces to PCA/MDS. As
emphasized in \citet{Purdom:2011}, the use of a non-identity $Q$ matrix incorporates
structure from known relationships between the $p$ taxa by exploiting a matrix
representation of phylogenetic relationships, thus providing a model for covariance
structure.

If we let $Q=LL'$ be a Cholesky decomposition of $Q$ and set $Z:=XL$, then the kernel
$K_Q$ has an SVD of the form $XQX'=VS^2V'$.  This leads to a two-dimensional regression
estimate that takes the same the form as $\hat\beta_{\mbox{\tiny PCR}}$ in
\eqref{eq:betaPCR}. Indeed, we can recover a primal space estimate in terms of singular
vectors as
\begin{equation}\label{eq:betaDPCR}
\hat\beta_{\mbox{\tiny DPCR}} := V_{(2)}S_{(2)}^{-1}U_{(2)}'y =
\sum_{k=1}^2\frac{1}{\sigma_k}u_k'y\,v_k.
\end{equation}
That is, implicit in a DPCoA plot is a coefficient vector, $\hat\beta_{\mbox{\tiny
DPCR}}$ which models a two-dimensional linear association between $y$ and $Z=XL$ in the
same way that $\hat\beta_{\mbox{\tiny PCR}}$  represents a two-dimensional linear
association between $y$ and $X$. Further, $XQX'=(XL)(XL)'$ and so $U$, $S$ and $V$ in
\eqref{eq:betaDPCR} are the same as those in the penalized (non-truncated) estimate,
$\hat\beta_{Q}$, in \eqref{eq:hatBetaQ}.

\subsection{Kernel-based regression with two kernels}\label{sec:KPR}
In addition to similarities among taxa, as in $Q$, it is often of interest to incorporate
similarities among samples as derived, for instance, from UniFrac distances:
$H=-\frac{1}{2}\calJ\Delta^U\calJ$. The symmetric positive definite $n\times n$ kernel
$H$ defines a new inner product on $\bbR^n$ given by $\langle u,w\rangle_H = u'Hw$, with
the corresponding norm $\|u\|^2_{H}=\langle u,u\rangle_H$. If we consider both a general
kernel, $H$, and a DPCoA kernel $K_Q=XQX'$, the generalized ridge estimate $\hat\beta_Q$
in \eqref{eq:hatBetaQ} can be extended to
\begin{equation}\label{eq:hatBetaQH}
\begin{aligned}
\hat\beta_{Q,H} & :=\argmin_{\beta\in\bbR^{p}}\left\{ \|y-X\beta\|_{H}^2 + \lambda\|\beta\|_{Q^{-1}}^2\right\}\\
    & = (X'HX + \lambda Q^{-1})^{-1}X'Hy.
\end{aligned}
\end{equation}

In this section, we show that the estimate in \eqref{eq:hatBetaQH} is directly defined
based on the generalized eigenvectors of the two kernels $Q$ and $H$. Before proceeding
to the general case, let us examine the special case of ridge regression. In this case,
$H = I_n$ and $Q=I_p$. It is well known that ridge estimates can be obtained by solving
an equivalent optimization problem in the dual space $\bbR^n$, known as \emph{kernel
ridge regression} \citep{Scholkopf:2002}. Specifically, taking $K_I=XX'$, the ridge
estimate in \eqref{eq:betaRidgeSVE} can be obtained as $\hat\beta_{\mbox{\tiny
ridge}}=X'\hat\gamma_{\mbox{\tiny kernel ridge}}$, where
\begin{equation}\label{eq:KernelRidge}
\begin{aligned}
\hat\gamma_{\mbox{\tiny kernel ridge}}& =(K_I+\lambda I)^{-1}y = (K_I^2+\lambda K_I)^{-1}K_Iy\\
     &= \argmin_{\gamma\in\bbR^{n}}\{ \|y-K_I\gamma\|^2 + \lambda\|\gamma\|^2_{K_I}\}.
\end{aligned}
\end{equation}

In the case of ridge, the connection between the dual- and primal-space estimates,
$\hat\gamma_{\mbox{\tiny kernel ridge}}$ and $\hat\beta_{\mbox{\tiny ridge}}$, relies on
the form $K_I=XX'$. Unfortunately, it is less clear how to extend this connection to a
general kernel (e.g., UniFrac or polynomial). One way to incorporate a general kernel $K$
and a second kernel $H$ in \eqref{eq:KernelRidge} is to define the penalty in terms of
$H$ as
\begin{equation}\label{eq:K2Hinv}
\hat\gamma_{*} = (K^2+\lambda H^{-1})^{-1}Ky = \argmin_{\gamma}\{ \|y-K\gamma\|^2 + \lambda\|\gamma\|_{H^{-1}}^2\},
\end{equation}
which is exactly the Tikhonov regularization, but in the dual space; compare
eq.~\eqref{eq:hatBetaQ}. However, $\hat\gamma_{*}\in\bbR^n$ has no obvious connection to
a penalized estimate of $\beta\in\bbR^p$ and cannot be used to obtain a penalized
regression estimate in the primal space, even if $K=K_I=XX'$.

To bridge this gap, we instead apply the Franklin regularization scheme
\citep{Franklin:1978}, a little-used alternative to Tikhonov regularization.  More
specifically, for {\em any} kernels $K$ and $H$, we \emph{define} the dual estimate
\begin{equation}\label{eq:Franklin}
\hat\gamma_{H,K} := (K+\lambda H^{-1})^{-1}y  = \argmin_{\gamma\in\bbR^{n}}\{ \|y-K\gamma\|_{K^{-1}}^2 + \lambda\|\gamma\|_{H^{-1}}^2\}
\end{equation}

Comparing \eqref{eq:K2Hinv} and \eqref{eq:Franklin}, one sees that the analytic form of
\eqref{eq:Franklin} involves just $K$ rather than $K^2=K'K$. As shown in
Proposition~\ref{prop:KPR}, this subtle difference is a key for relating
$\hat\gamma_{H,K}$ and its primal-space counterpart. Before presenting the main result of
this section, we provide several equivalent forms of $\hat\gamma_{H,K}$.
\begin{equation}\label{eq:Franklin2}
  \begin{aligned}
\hat\gamma_{H,K} &=(K+\lambda H^{-1})^{-1}y\\
    &= \argmin_{\gamma\in\bbR^{n}}\{ \|y-K\gamma\|_{K^{-1}}^2 + \lambda\|\gamma\|_{H^{-1}}^2\}\\
    &=\argmin_{\gamma\in\bbR^{n}}\{\|y-K\gamma\|_{H}^2+ \lambda \|\gamma\|_{K}^2\}\\
    &=(HK+\lambda I)^{-1}Hy\\
    &=H(KH+\lambda I)^{-1}y.
    \end{aligned}
\end{equation}

In Proposition~\ref{prop:KPR}, we also refer to the special case corresponding to the
DPCoA ordination. As before, let $Z=XL$ so that $K_Q=XQX'=XLL'X'=ZZ'$.  Taking $H=I$, the
dual-space estimate in \eqref{eq:Franklin} is $\hat\gamma_{I,K_Q} = (K_Q+\lambda
I)^{-1}y$, and so the corresponding primal-space estimate is $\hat\beta\equiv
Z'\hat\gamma_{I,K_Q}$. Since this estimate arises from the DPCoA kernel, so we make the
following definition.

\begin{definition}\label{def:betaDPCoA}
A primal space DPCoA estimate is of the form $\hat\beta_{\mbox{\tiny
DPCoA}}=Z'\hat\gamma_{I,K_Q}=L'X'(XQX'+\lambda I)^{-1}y$.
\end{definition}

The next proposition collects several properties that emphasize the roles of $H$ and $K$
in our penalized regression framework. In particular, we show that the primal space
estimate $\hat\beta_{Q,H}$ can be recovered in terms of two kernels, $H$ and $K_Q$.

\begin{proposition}\label{prop:KPR}
Let $H$ and $K$ be any two kernels constructed using the rows of $X$ in the regression
model $y=X\beta+\epsilon$.  Then,
\begin{enumerate}
\item  $\hat\gamma_{H,K}$ is a linear combination of the eigenvectors of the matrix
    product $HK$.

\item For any kernel $H$ and DPCoA kernel $K_Q=XQX'$, then the  primal- and dual-space
    estimates in \eqref{eq:hatBetaQH} and \eqref{eq:Franklin}, respectively,  are
    related as: $\hat\beta_{Q,H} = QX'\hat\gamma_{H,K_Q}$.

\item For $H=I$ and $Q=LL'$, the generalized ridge and DPCoA estimates  are related as
    $\hat\beta_{Q} = QX'(K_Q+\lambda I_n)^{-1}y = L\hat\beta_{\mbox{\tiny DPCoA}}$.



\end{enumerate}
\end{proposition}

\begin{proof}  These relationships are based on several basic linear algebraic
properties.  In particular, we make use of the following identities:
\begin{equation}\label{eq:XQHidentities}
\begin{aligned}
&(X'HX + \lambda Q^{-1})^{-1}X'H  =(QX'HX+\lambda I_p)^{-1}QX'H \\
                              & = Q(X'HXQ+\lambda I_p)^{-1}X'H
                              = QX'(XQX'+\lambda H^{-1})^{-1} \\
                             & = QX'H (XQX'H+\lambda I_n)^{-1}.
 \end{aligned}
\end{equation}

\begin{enumerate}
\item Let $w_1,...,w_n$ denote the columns of the matrix $W$ satisfying $W'KW =
    \diag\{\sigma_1,...,\sigma_n\}$ and $W'H^{-1}W = \diag\{\mu_1,...,\mu_n\}$. That
    is, $\{w_k\}$ is a basis with respect to which a simultaneous diagonalization of
    $K$ and $H^{-1}$ is obtained (see \citet{GolubVanLoan:2012,HansenBook:1998}). These
    are the generalized eigenvectors of the pair $(K,H^{-1})$. Then
\begin{equation}\label{eq:gammaH_GSV}
\hat\gamma_{H,K} = \sum_{k=1}^n \left(\frac{\sigma_k}{\sigma_k+\lambda\mu_k}\right)     
 \frac{w_k'y}{\sigma_k}  w_k.
\end{equation}
Since $H$ is invertible, the $w_k$'s are also eigenvectors of $HK$
\citep{GolubVanLoan:2012}.
\item Using the identities in \eqref{eq:XQHidentities}, the estimate in
    \eqref{eq:hatBetaQH} can be expressed as
\begin{equation}\label{eq:hatBetaQH2}
\begin{aligned}
 \hat\beta_{Q,H} & = Q(X'HXQ+\lambda I_p)^{-1}X'Hy \\
                & = QX'(XQX'+\lambda H^{-1})^{-1}y \\
                & = QX'(K_Q+\lambda H^{-1})^{-1}y = QX'\hat\gamma_{H,K_Q}.
\end{aligned}
\end{equation}

\item Setting $H=I$ in the first line of the above equalities gives $\hat\beta_{Q} =
    QX'(XQX'+\lambda I_n)^{-1}y$.  Using $Q=LL'$ gives
    $\hat\beta_{Q}=L\,L'X'(K_Q+\lambda I_n)^{-1}y = L\,Z'(ZZ'+\lambda I)^{-1}y=
    L\hat\beta_{\mbox{\tiny DPCoA}}$.
\end{enumerate}
\end{proof}

\noindent\textbf{Remarks:}
\begin{enumerate}\renewcommand{\labelenumi}{\Alph{enumi})}
\item \textbf{Types of similarity kernels.} In general, a sufficient condition for a
    matrix $K$ to be a similarity kernel is that it is induced by a feature map
    $\phi\colon \bbR^p \to \mathcal{K}$. More specifically, the $i,j$ entry of $K$ is
    defined as the inner product of the observations $x_i\in \Bbb{R}^p$ with respect to
    their transformed versions $K_{ij}=\langle\phi(x_i),\phi(x_j)\rangle$ in the new
    inner product space, $(\mathcal{K},\langle\cdot,\cdot\rangle)$. Examples include
    $K_I=XX'$ or $K_Q=XQX'$, where $\mathcal{K}$ is $\bbR^p$ with inner product
    $\langle\cdot,\cdot\rangle_Q$ (as in DPCoA). It is this quadratic form that we
    require for $K_Q$ in Proposition~\ref{prop:KPR}(2)--(3); see \citet{Freytag:2014}
    for genomic applications of this form. On the other hand, $H$ can be any symmetric
    positive semi-definite matrix. Here, we are more interested in
    biologically-motivated kernels, such as UniFrac or DPCoA, than
    mathematically-derived ones, such as those constructed from polynomials or radial
    basis functions \citep{Scholkopf:2002}.

\item \textbf{Co-informative kernels and the HSIC.} Any kernels $K$ and $H$ may be used
    in \eqref{eq:Franklin} and \eqref{eq:Franklin2}, but to be useful in this
    framework, we assume that they are ``co-informative" in the sense that they exhibit
    a shared eigenstructure; for instance, both should be informative for classifying
    samples. This concept is illustrated in the simulation of Section~\ref{ex:EdgePCA}
    and Figure~\ref{fig:EPCA}. The co-informativeness can be made precise using the
    Hilbert-Schmidt information criteria (HSIC) \citep{Gretton:2005} 
    or its relatives the distance covariance \citep{Szekely:2009} 
    or RV statistic \citep{Robert-EscoufierRV:1976}. 
    \citet{JosseHolmes:2013} provide a nice review of these and related kernel-based
    tests. The HSIC provides a test for the statistical dependence of two data sets,
    $X_1$ ($n\times p$) and $X_2$ ($n\times q$), and is based on the eigen-spectrum of
    covariance operators defined by kernels defined by $X_1$ and $X_2$, respectively.
    For two kernels $K$ and $H$, the empirical HSIC is simply $\trace(HK)$. The HSIC is
    thus of particular interest in item (1) of Proposition~\ref{prop:KPR}, which shows
    how two co-informative kernels may be used to obtain a penalized estimate
    $\hat\beta_{Q,H}$.

\item \textbf{Linear mixed models and KPR.} As an alternative to the regularization
    framework presented here, it may be useful to consider a kernel as a generalized
    covariance among either the $p$ variables (using $Q$) or $n$ subjects (using $H$)
    \citep{Purdom:2011,Schaid:2010a}. This alternative representation can be made
    precise using the linear mixed model (LMM) framework \citep{RuppertWandCarroll:03}.
    Specifically, recall from equations \eqref{eq:hatBetaQH} and \eqref{eq:Franklin2}
    that
\begin{equation*}
\begin{aligned}
\hat\beta_{Q,H} & = \argmin_{\beta\in\bbR^{p}}\{ \|y-X\beta\|_{H}^2 + \lambda\|\beta\|_{Q^{-1}}^2\}\\
                & = QX'(K_Q+\lambda H^{-1})^{-1}y \\
                & = QX'\argmin_{\gamma\in\bbR^{n}}\{ \|y-K_Q\gamma\|_{H}^2 + \lambda\|\gamma\|_{K_Q}^2\}
                 = QX'\hat\gamma_{H,K_Q}.
\end{aligned}
\end{equation*}
These regression estimates are compatible with $\beta\sim N(0,\sigma_b^2 Q)$, \
$\epsilon\sim N(0,\sigma_e^2 H^{-1})$ and $var(y) = (\tau K_Q + \lambda H^{-1})^{-1}$.
And the estimate $\hat\gamma_{H,K_Q}$ is compatible with $\gamma\sim N(0,\sigma_a^2
K_Q^{-1})$ and $\epsilon\sim N(0,\sigma_e^2 H^{-1})$. With regard to the latter, a
genetic similarity between subjects (e.g., kinship) is often used for grouping subjects
and several authors have proposed this form of kernel for testing the (global) genetic
association with a trait or phenotype, $y$; see, e.g., \citet{Schifano:2012}. 
In particular, these methods use the LMM framework to motivate and define a ``kernel
association test".  The variance score statistic for testing the null hypothesis of no
association between $y$ and $X$ ($H_0: \beta=0$) is, using our notation above,
$\mathcal{T} := \|y\|_{H^{1/2}K_QH^{1/2}}^2$. The kernel association testing framework
has been applied to microbiome data using a single kernel at a time derived from
Unifrac \citep{NiZhao:2015}, but this is a test for whether $\beta\ne 0$ and, unlike
our KPR framework, provides no insight about which taxa, as represented by coordinates
of $\beta$, are associated with $y$.
\end{enumerate}

\subsection{Regression with compositional data}\label{sec:compdata}
Data from 16S rRNA gene sequencing methods are random counts of the molecules in each
sample. The number of sequence reads assigned to a taxon contains no information about
the {\em actual} number of molecules in the sample; the total number of reads observed in
two samples can vary by several orders of magnitude. Hence, only \emph{relative} amounts
can be investigated. Common approaches for normalizing these data include converting them
to proportions (relative percent) or subsampling the sequences to create equal library
sizes for each sample (rarefying). These data are ``compositional" in the sense that the
microbial abundances represent a proportion of a constant total. It is well known,
however, that compositional measures can result in spurious correlations among taxa
\citep{Pearson:1896,Aitchison:2003,FriedmanAlm:2012}, an effect that can be quite extreme
when there are a few dominant taxa.

Compositional data reside on the \emph{simplex} $\bbS^{p-1}$ of unit-sum vectors in
$\bbR^p$ and so standard multivariate methods do not apply
\citep{AitchisonBook:2003,Egozcue-Chap2:2011,Li:2014,Lovell:2015}. In particular, because
these data do not naturally reside in a Euclidean vector space, standard regression
models based on Euclidean covariance measures are inappropriate. However, ordinary
least-squares and ridge regression estimates are of the form $\hat\beta= ( X'X + \lambda
I)^{-1} X' y$ (with $\lambda = 0$ and $\lambda > 0$, respectively). Thus, these estimates
depend on the empirical covariance structure, $X'X$, among taxa, which may include
spurious correlations. Similarly, \citet{Li:2014} points out that a na\"ive application
of lasso regression is not expected to perform well due to the compositional nature of
the covariates. He addresses this issue by applying a lasso regression model to the
log-ratio abundances and imposing an additional constant-sum constraint on the
coefficient vector, $\beta$.

We next show that the generality of KPR for handling non-Euclidean structures can be used
to address the compositional nature of microbiome data. In particular, we propose an
approach that uses the centered log-ratio transformation of the compositional vectors and
an estimate of covariance among the log taxa counts that is obtained via Aitchison's
variance matrix \citep{AitchisonBook:2003,Egozcue-Chap2:2011}.

Let $X$ be the $n\times p$ sample-by-taxon matrix whose rows are relative percent
(compositional) vectors $\{x_i\}_{i=1}^n\subset\bbS^{p-1}$. The columns of $X$ will be
denoted by $x^k$, corresponding to $k=1,...,p$ taxa.  Let $g(z) =\left(\Pi_{k=1}^p
z^k\right)^{1/p}$ be the geometric mean of a row vector, $z$, and denote the centered
log-ratio (CLR) transform of $x_i$ by $\tilde x_i=\clr(x_i):=
\left[\log\frac{x_i^1}{g(x_i)},\ldots,\log\frac{x_i^p}{g(x_i)}\right]$. In what follows
we denote the matrix of CLR vectors by $\tilde X$, and use the normalized variation
matrix $T$, of $X$, as defined by \citet{Aitchison:1982}: $T_{k,\ell}=
\var\left(\frac{1}{\sqrt{2}}\log\frac{x^k}{x^\ell}\right)$.  $T$ is a symmetric
dissimilarity matrix with zeros on the diagonal and entries that have squared Aitchison
distance units: the Aitchison norm of a vector $x\in\bbS^{p-1}$ is defined as $\|x\|_a^2
= \frac{1}{2p}\sum_{k,\ell}\left(\log \frac{x^k}{x^\ell} \right)^2$. In fact, $\|x\|_a^2
= \|\clr(x)\|^2$.  One can show that $T$ is related to the covariance matrix, $C$ of the
log of the true unobserved taxa counts via $T = v\boldsymbol{1}' + \boldsymbol{1}v' - 2C$
\citep{Li:2014}. Consequently, $C = -\frac{1}{2}\mathcal{J}T\mathcal{J}$, and we can use
$C$ in place of $Q$ in eq.~\eqref{eq:hatBetaQ} to obtain
\begin{equation}\label{eq:hatBetaC}
 \tilde\beta_C =\argmin_{\beta}\left\{ \|y-\tilde{X}\beta\|_{\bbR^n}^2 + \lambda\|\beta\|_{C^{-1}}^2\right\}.
\end{equation}

As a comparison, we observe that \citet{Li:2014} proposed a constrained regression
\begin{equation}\label{eq:LiConstrainReg}
\E(y_i) = \beta_1\log x_i^1 +\dots + \beta_p\log x_i^p \quad \text{ subject to } \sum_{j=1}^p\beta_j=0,
\end{equation}
augmented with a lasso penalty to obtain an estimate of the form
\begin{equation*}
\argmin_{\beta} \left\{\frac{1}{2n} \| y- \sum_{j} \log(x^j) \beta_j \|_{\bbR^n}^2 + \lambda \sum_j | \beta_j | \right\}  \quad \mbox{ subject to } \sum_{j=1}^p\beta_j=0.
\end{equation*}
The zero-sum constraint on $\beta$ was emphasized for interpretability advantages over
the standard lasso estimate. Temporarily denoting $\beta_p=-\sum_{j = 1}^{p-1}\beta_j$,
we see that \eqref{eq:LiConstrainReg} is equivalent to
\begin{align*}
\E(y_i) &=\beta_1\log\frac{x_i^1}{x_i^p} + \beta_2\log\frac{x_i^2}{x_i^p} + \dots + \beta_{p-1}\log\frac{x_i^{p-1}}{x_i^p} \\
&=\beta_1\log x_i^1 + \beta_2\log x_i^2  + \dots + \beta_{p-1}\log x_i^{p-1}  -\sum_{j = 1}^{p-1}\beta_j \cdot\log x_i^p.
\end{align*}
Since $\sum_{j = 1}^{p}\beta_j=0$, this can be rewritten as
\begin{align*}
\E(y_i) &= \beta_1\log x_i^1 +\dots +\beta_p \log x_i^p - (\sum_{j = 1}^p \beta_j)  \log g(x_i)\\
&=\beta_1\log \frac{x_i^1}{g(x_i)} + \dots + \beta_p\log\frac{x_i^p}{g(x_i)}  \qquad \text{ subject to } \sum_{j=1}^p\beta_j=0.
\end{align*}
Therefore, Li's proposal of regression on log-ratio abundances is equivalent to
regression on the CLR-transformed data $\tilde{X}$ provided a zero-sum constraint is
imposed on $\beta$. In contrast, however, our formulation does not explicitly impose a
constant-sum constraint. In fact, this constraint is not needed because the CLR transform
removes the analysis from the simplex to allow an analysis in Euclidean vector space
algebra \citep{Egozcue-Chap2:2011}. Our model instead incorporates the appropriate
covariance structure for the CLR transformation, $C$.

As a final observation, we note that a positive-definite $C$ in \eqref{eq:hatBetaC}, or
more generally $Q$ in \eqref{eq:hatBetaQ}, can be decomposed as a sum $Q=I+\tilde Q$ of
the identity plus a positive semi-definite singular matrix $\tilde Q$. The identity term
constrains $\sum_{j=1}^p\beta^2_j$ to be small while, overall, $\tilde Q$ encourages
extrinsic structure (e.g., smoothness). One may also control the size of
$\sum_{j=1}^p\beta_j^2$ by adding or subtracting values in the diagonal entries of $Q$.
This idea is similar to that of ``Grace-ridge" in \citet{ZhaoShojaie:2015} where, in
addition to the penalty induced by $Q$, the authors propose to further impose a
ridge-type penalty in the objective.  We apply the significance testing framework of
\citet{ZhaoShojaie:2015} in Section~\ref{sec:Applic}.

\section{Numerical Experiments}\label{sec:Examples}
To illustrate the proposed framework, we perform several data-driven simulations using
publicly available microbiome data.  We consider three scenarios from the literature that
exploit extrinsic structure from a phylogenetic tree, including DPCoA, UniFrac and edge
PCA. To achieve realistic simulations, we simulate ``true" signals of the type implied by
each of these methods in order to create benchmarks for performance evaluation. Our
emphasis is on formalizing the role that such structure plays in penalized regression
when modeling associations between the multivariate data, $X$, and a response variable,
$y$. Since $y$ is directly simulated from $X$ in these settings, the compositional nature
of the data discussed in Section~\ref{sec:compdata} does not affect the simulation
results. We will return to this topic when analyzing the relative abundance data in
Section~\ref{sec:Applic}.

The numerical experiments in this section are motivated by the relationship between the
PCoA plots and PCR described in Section~\ref{sec:PCoA-PCR} and
Figure~\ref{fig:PCoA-UniFrac-age}(b). This connection can be generalized to a number of
other commonly-used graphical representations in the microbiome literature. For instance,
any two-dimensional DPCoA plot involves an implicit coefficient vector, $\beta$, of
associations between $y$ and $X$.

Throughout this section, we compare the performance of KPR with ridge regression and
lasso. Ridge regression provides a direct extension of ordinary least squares and thus is
a natural benchmark for comparing various KPR estimates. Lasso, which gives sparse
estimates, is used as a benchmark in settings where the true $\beta$ is sparsely
non-zero. The choice of competing methods is limited by our emphasis on {\em estimating}
$\beta$, rather than predicting the outcome $y$.  Indeed, most kernel methods focus on
prediction which renders them inappropriate for comparison.

In all simulation experiments, the tuning parameters for KPR, ridge and lasso are chosen
using 10-fold cross-validation. Specifically, to compare the prediction performance of
KPR, ridge and lasso, we choose the tuning parameters that minimize squared test error in
held out cross validation samples (CV min). On the other hand, the task of estimation
usually requires more smoothing than prediction \citep{CaiHall:2006}. Therefore, when
examining the estimation performances of KPR, ridge and lasso, we use the largest tuning
parameters such that the squared test errors are within one standard error of the minimum
squared test error (CV 1se), as suggested in \citet{HasTibFri:2009}. For comparison, we
also consider the tuning parameters corresponding to the minimum squared test error for
ridge and lasso.

\subsection{Regression and DPCoA}\label{ex:DPCoA}
In our first example, we compare the estimation and prediction performances of KPR, ridge
and lasso using the data depicted in Figure~\ref{fig:PCoA-UniFrac-age}. The rows of $X$
represent relative abundances of $p = 149$ taxa from $n=100$ subjects in a study by
\citet{Yatsunenko:2012}. The outcome $y$ is log-transformed age of each subject. For KPR,
we use $K_Q=XQX'$ and $H=I$, where $Q=-\frac{1}{2}\calJ\delta\calJ$ is a matrix of
similarities between taxa obtained from the matrix of squared patristic distances,
$\delta$. Motivated by DPCoA plots, we assume the underlying ``true'' response
$y_{\mbox{\tiny True}}$ is generated from the first two eigenvectors of $K_Q$. Let $L$ be
the Cholesky factor of $Q$, i.e., $Q=LL'$, and let $XL=U^{L}S^{L}(V^{L})'$. Recall that
$A_{(k)}$ denotes the first $k$ columns of matrix $A$, or its first $k$ rows and columns
if $A$ is diagonal. Motivated by \eqref{eq:betaDPCR}, we let
\begin{equation}\label{eq:truebetaDPCoA}
  \beta_{\mbox{\tiny True}}= s\left(V^{L}_{(2)}(S^{L}_{(2)})^{-1}(U^{L}_{(2)})'y, \tau \right),
\end{equation}
where, $s(\cdot,\tau)$ is the hard-thresholding operator, i.e., $s(x, \tau) =
x\cdot1(|x|>\tau)$. The threshold $\tau\geq 0$ is set to achieve various levels of
sparsity:  $\|\beta_{\mbox{\tiny True}}\|_0\in\{\lfloor0.2p\rfloor, \lfloor0.6p\rfloor,
p\}$. After generating $\beta_{\mbox{\tiny True}}$, we simulate
\begin{equation*}
y_{\mbox{\tiny True}}=U^L_{(2)}S^L_{(2)}(V^L_{(2)})'\beta_{\mbox{\tiny True}}.
\end{equation*}


The simulation is repeated 500 times, each with a different $\epsilon\sim N_n(0,
\sigma^2_\epsilon I_n)$ in $y_{obs} = y_{\mbox{\tiny True}} + \epsilon$. Further,
$\sigma^2_\epsilon$ is set to achieve $R^2 = \var(y_{\mbox{\tiny True}}) /
(\var(y_{\mbox{\tiny True}}) + \sigma^2_\epsilon)\in\{0.1,0.2,...,0.9\}$. In each
repetition, we estimate $\hat\beta_{\mbox{\tiny DPCoA}}$ from $y_{obs}$ according to
definition \ref{def:betaDPCoA}. To make the simulation more realistic, we do not assume
we always observe the $Q$ matrix used to simulate $\beta_{\mbox{\tiny True}}$ and
$y_{\mbox{\tiny True}}$. Rather, to estimate $\hat\beta_{\mbox{\tiny DPCoA}}$, we use
$Q_{obs}$, which is obtained by adding random Gaussian noise to $Q$. Eigenvalues of
$Q_{obs}$ are adjusted to be equal to the eigenvalues of $Q$.
The amount of Gaussian noise added to the entries of $Q_{obs}$ is empirically determined
to achieve $\|Q-Q_{obs}\|_F/\|Q\|_F\in\{0, 0.25, 0.5\}$. As a comparison, we estimate
$\hat\beta_{\mbox{\tiny Ridge}}$ and $\hat\beta_{\mbox{\tiny Lasso}}$ using only $X$ and
$y_{obs}$, without incorporating $Q$. From the estimated coefficients, we compute $\hat
y_{\mbox{\tiny DPCoA}}=XL_{obs}\hat\beta_{\mbox{\tiny DPCoA}}$, $\hat y_{\mbox{\tiny
Ridge}}=X\hat\beta_{\mbox{\tiny Ridge}}$ and $\hat y_{\mbox{\tiny
Lasso}}=X\hat\beta_{\mbox{\tiny Lasso}}$.
The performance metrics are the prediction sum of squared error (PSSE) from
$y_{\mbox{\tiny True}}$ and estimation sum squared error (ESSE) from $\beta_{\mbox{\tiny
True}}$.

\begin{figure}
\vspace{-5pt}
\centering
{\includegraphics[height = 17 cm]{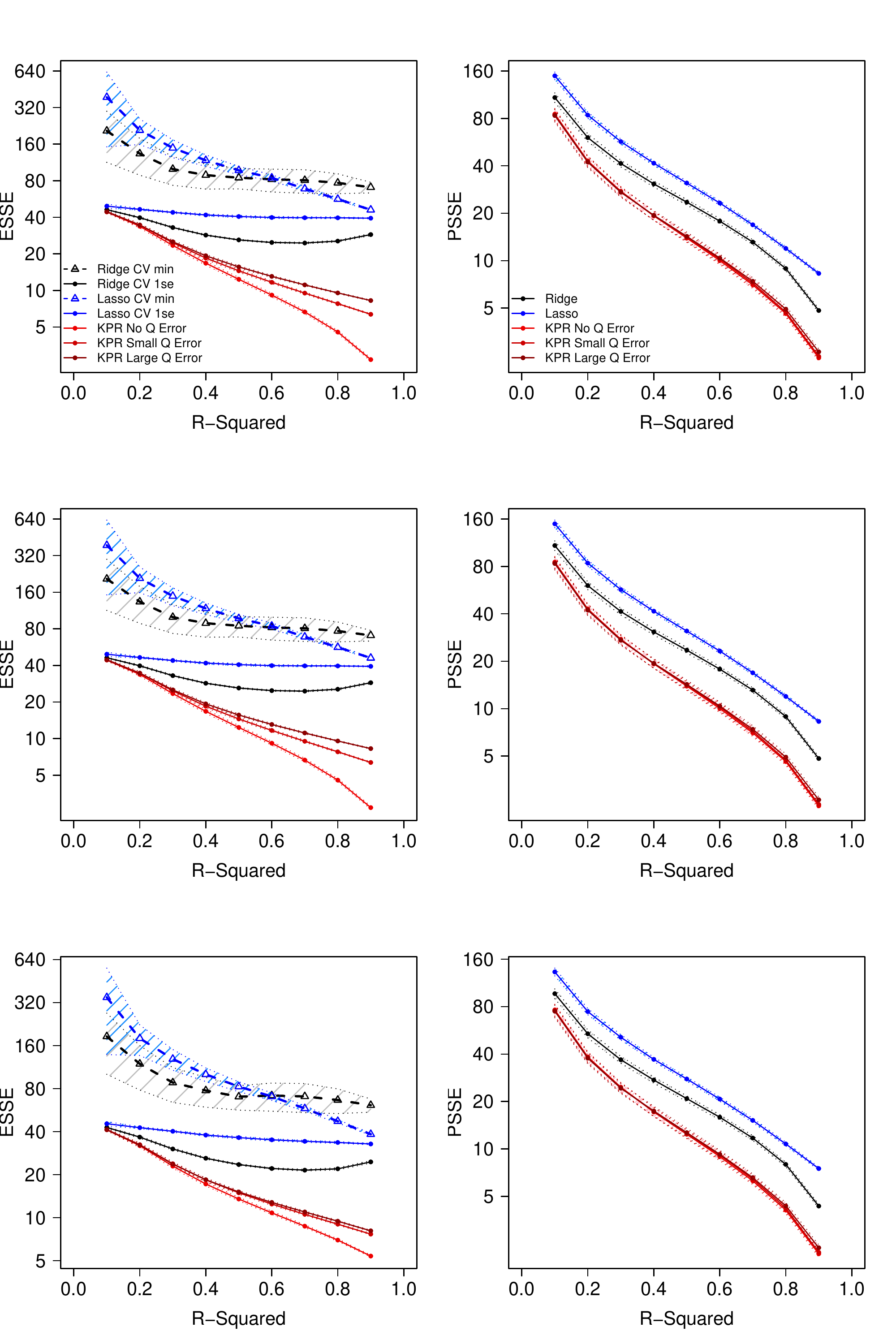}}
\vspace{-5pt}\caption{\small Estimation sum squared error (ESSE: left panels) and prediction
sum squared errors (PSSE: right panels) of KPR (red), ridge regression (black) and lasso (blue),
and their 95\% confidence bands. We consider three sparsity settings for
$\beta_{\mbox{\tiny True}}$, based on \eqref{eq:truebetaDPCoA}: $\|\beta_{\mbox{\tiny True}}\|_0 = p$
in top panels; $\|\beta_{\mbox{\tiny True}}\|_0 = \lfloor0.6p\rfloor$ in center panels, and
$\|\beta_{\mbox{\tiny True}}\|_0 = \lfloor0.2p\rfloor$ in bottom panels.
For ridge and lasso, tuning parameters that produce the smallest cross-validated squared test
error (CV min), and the largest tuning parameters such that the cross-validated squared test
errors are within one standard error of the minimum cross-validated squared test error (CV 1se)
are considered. For KPR, we consider $\|Q-Q_{obs}\|_F/\|Q\|_F = 0$ (no $Q$ error), $0.25$
(small $Q$ error) and $0.5$ (large $Q$ error). }
\label{fig:DPCoASim}
\end{figure}

Figure~\ref{fig:DPCoASim} shows the estimation and prediction performance of KPR, ridge
and lasso. KPR significantly outperforms both ridge regression and lasso for both
prediction and estimation in all settings. As expected, the performance of ridge and
lasso for \emph{estimation} improve when using a larger tuning parameter. On the other
hand, neither mis-specification of $Q$ nor sparsity of $\beta_{\mbox{\tiny True}}$ seems
to substantially impact the relative performance of the three methods. This may be due to
the fact that KPR estimates the correct target $\beta_{\mbox{\tiny True}}$, even with
mis-specified $Q$, whereas ridge regression and lasso estimate the wrong target.

\subsection{Regression and PCoA with respect to a UniFrac kernel}\label{ex:UniFrac}
In the case of PCoA with respect to a UniFrac matrix $\Delta^U$ of squared
dissimilarities, the graphical displays are based on the eigen-decomposition of
$H=-\frac{1}{2}\calJ \Delta^U\calJ$.  That is, for $H = U^H (S^H)^2 (U^H)' \approx
U^{H}_{(2)}(S^{H}_{(2)})^2(U^{H}_{(2)})'$, the $n$ samples are represented in two
dimensions by the columns of $U^{H}_{(2)}S^{H}_{(2)}$;  this results in points
$\{\eta_{i1}, \eta_{i2}\}_{i=1}^n :=\{(\sigma_1 U^{H}_{i1}, \sigma_2
U^{H}_{i2})\}_{i=1}^n$, as plotted in Figure~\ref{fig:PCoA-UniFrac-age}. When the points
are colored according to a response variable, $\{y_i\}_{i=1}^n$, the implied regression
model is
\begin{equation}\label{eq:KPCR-model}
  \begin{aligned}
     y &= \gamma_1\eta_{1} + \gamma_2\eta_{2}+\epsilon \\
       &= U^H_{(2)}S^H_{(2)}\,\gamma + \epsilon.
  \end{aligned}
\end{equation}
However, in contrast to PCR in eq.~\eqref{eq:gammaPCR}, where $US=XV$, it is not obvious
how to connect $\gamma$ directly to the $p$-coordinates corresponding to the $p$ columns
of $X$.
Here, we exploit the joint eigenstructure of kernels $K_I$ and $H$ (see
\eqref{eq:gammaH_GSV}) by proceeding as in \eqref{eq:Franklin2} to obtain the estimate
$\hat\beta_{H}=X'\hat\gamma$ as in \eqref{eq:hatBetaQH2}, with $Q=I$.

In this example, we use the same data as in Section \ref{ex:DPCoA}. For KPR, we use
$K=XX'$ and obtain $H=-\frac{1}{2}\calJ \Delta^U\calJ$ using the UniFrac distance matrix.
We simulate $\gamma_{\mbox{\tiny True}}$ and $y_{\mbox{\tiny True}}$ from the first two
eigenvectors of $H$, as in \eqref{eq:KPCR-model}:
\begin{align*}
\gamma_{\mbox{\tiny True}}&=\left((U^{H}_{(2)})'(S^{H}_{(2)})^2U^{H}_{(2)}\right)^{-1}S^{H}_{(2)}(U^{H}_{(2)})'y \\
y_{\mbox{\tiny True}}&=U^{H}_{(2)}S^{H}_{(2)}\gamma_{\mbox{\tiny True}}.
\end{align*}
This bivariate regression is illustrated in Figure~\ref{fig:PCoA-UniFrac-age}(b).

The simulation is repeated 500 times, each with a different $\epsilon\sim N_n(0,
\sigma^2_\epsilon I_n)$ to produce various values of $R^2\in\{0.1,0.2,\ldots,0.9\}$. We
compute $\hat y_{\mbox{\tiny KPR}}=K\hat\gamma_{\mbox{\tiny KPR}}$, where
$\hat\gamma_{\mbox{\tiny KPR}}$ is estimated using \eqref{eq:Franklin}. Similar to the
last example, we do not assume we always observe the $H$ matrix that is used to generate
$\gamma_{\mbox{\tiny True}}$ and $y_{\mbox{\tiny True}}$; rather, we use a noisy version,
$H_{obs}$, of $H$ in KPR with $\|H-H_{obs}\|_F/\|H\|_F\in\{0, 0.25, 0.5\}$.

Since there is no meaningful way to simulate $\beta_{\mbox{\tiny True}}$, we do not
compare the methods based on their estimation performances, and only consider prediction.
For all three methods, we find the tuning parameters that minimize the cross-validated
$H_{obs}$-weighted squared test error. While the use of $H$ in tuning ridge and lasso
penalties deviates from the common practice, it results in improved performances, given
the important role of $H$ in this simulation. The $H$ matrix also defines the valid
distance in this example. Thus, to evaluate the prediction performances of various
methods, we use the $H$-weighted prediction sum of squared error (HPSSE), $\|\hat
y-y_{\mbox{\tiny True}}\|^2_H$.

Figure~\ref{fig:UniFracSim} shows that KPR consistently outperforms ridge regression and
lasso in prediction, even with a reasonable amount of misspecification of $H$. This may
be due to the fact that, with the incorporation of the $H$ matrix, KPR estimates the
correct target whereas ridge and lasso do not.

\begin{figure}
\centering
{\includegraphics[height = 6.5 cm]{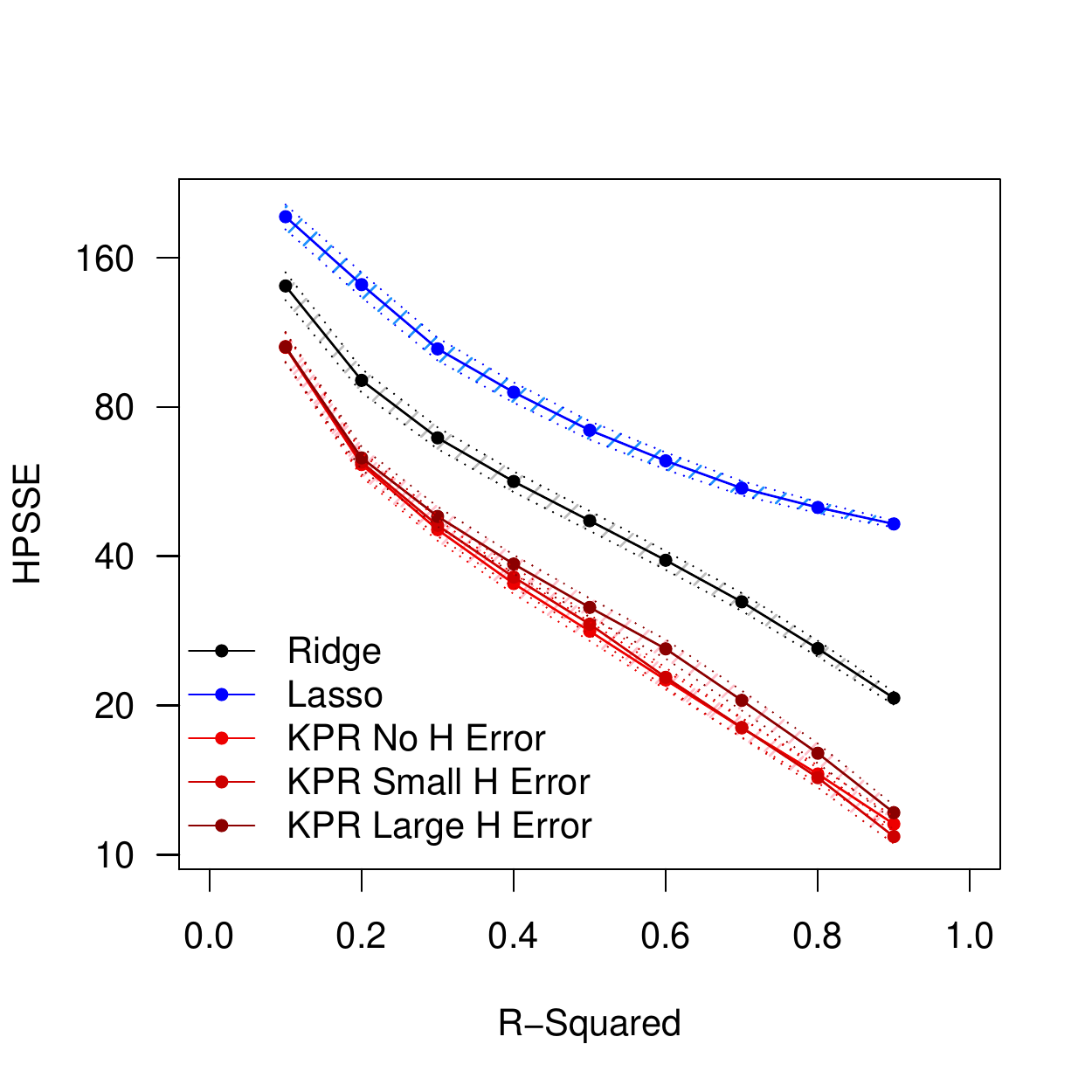}}
\vspace{-7pt}\caption{\small $H$-weighted prediction sum of squared error (HPSSE) of KPR (red), ridge (black) and lasso (blue), with 95\% confidence bands. For KPR, we consider
$\|H-H_{obs}\|_F/\|H\|_F = 0$ (no $H$ Error), $0.25$ (small $H$ Error) and $0.5$ (large $H$ Error). }
\label{fig:UniFracSim}
\end{figure}

\begin{figure}
\centering
{\includegraphics[width=.57\textwidth]{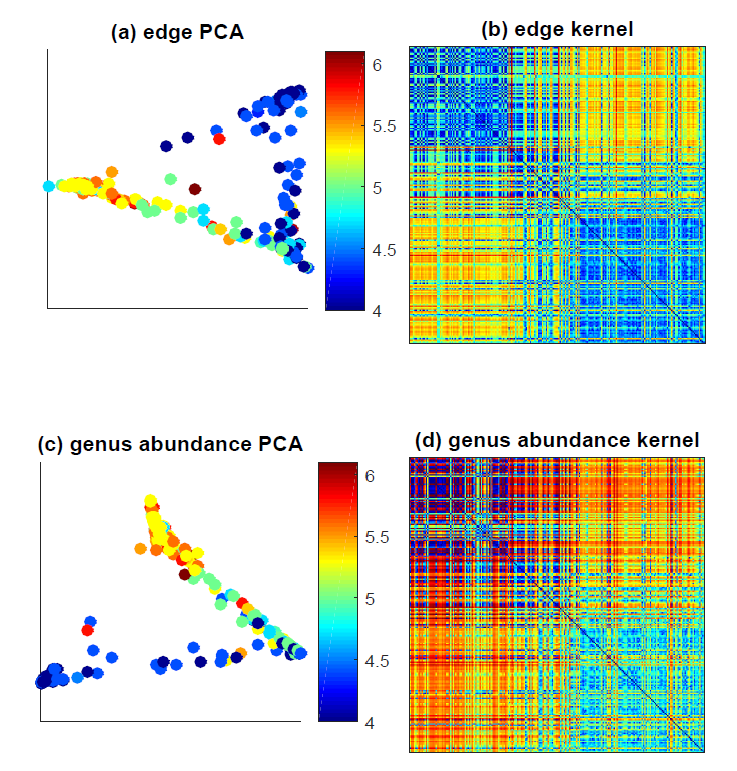}}
\vspace{-7pt}\caption{\small \textbf{Analysis of bacterial vaginosis data from \citet{Srinivasan:2012}.}
\textbf{(a)}: representation of the samples in the space of the first two PC's
    of the edge-matrix kernel $H=EE'$. The color of each point corresponds to the pH
    level of the sample;
\textbf{(b)}: heatmap of edge-matrix kernel used to generate the plot in (a);
\textbf{(c)}: two-dimensional PCA plot based on the genus-level relative abundances;
\textbf{(d)}: heatmap of the genus-abundance kernel $K=XX'$ used to create the plot in (c).
    In (b) and (d), subjects are ordered by the pH values.
}
\label{fig:EPCA}
\end{figure}

\subsection{Regression and PCoA using an edge-matrix kernel}\label{ex:EdgePCA}
In this section, simulations are based on data from a study of bacterial vaginosis (BV)
by \citet{Srinivasan:2012} in which 16S rRNA gene samples were collected using vaginal
swabs from $n=220$ women with and without BV. Here, the outcome $y$ represents pH
measured from vaginal fluid of each subject and we consider the association of $y$ with
genus-level taxa. In this example, we use the $p=62$ genera that exhibit non-zero
sequence counts in at least 20\% of the subjects.  So here, $X$ represents $220 \times
62$ abundances in a sample-by-genus matrix, and we use a kernel $K=XX'$. Additionally,
however, we define a second kernel $H=EE'$ based on the ``edge mass difference matrix",
$E$, originally introduced by \citet{Matsen:2013}. If the full phylogenetic tree has $q$
edges, each sample can be represented by a vector indexed by all $q$ edges, the
$e^{\text{th}}$ coordinate of which quantifies the difference between the fraction of
sequence reads on either side of the edge; i.e., the fraction of reads observed on the
root side of the tree minus the fraction of reads on the non-root side. We refer
to~\citet{Matsen:2013} for details and a discussion of ``edge PCA", which refers to PCA
applied to the $n\times q$ matrix $E$. Note, in particular, that abundances from
\emph{every} taxon level in the tree contribute to a similarity between subjects as
opposed to abundances at a single taxon level, which is used in UniFrac or DPCoA.

In summary, $X$ represents $p=62$ genus-level abundances while $E$ is based on all
$q=1770$ edges in the original phylogenetic tree. Figure~\ref{fig:EPCA}(a) shows a PCA
plot of the 220 subjects in which their similarity is defined using the edge kernel
$H=EE'$; the color of each dot represents the subject's pH.  Figure~\ref{fig:EPCA}(b) is
a heatmap of the kernel $H$ used to create Figure~\ref{fig:EPCA}(a). The columns and rows
of $H$ represent similarities between samples based on the edge mass difference matrix,
ordered by subject pH measurement. Similarly, Figure~\ref{fig:EPCA}(c) is a PCA plot
based on similarities defined using the genus-level abundance kernel, $K=XX'$.
Figure~\ref{fig:EPCA}(d) is a heatmap of the kernel $K$ used to create
Figure~\ref{fig:EPCA}(c), and subjects are again ordered by pH. These figures illustrate
how two different measures of similarity (two separate kernels) may be co-informative in
the sense that they both provide information about grouping of subjects' microbiota in
relation to their pH. It is thus natural to expect that incorporating information from
both $H$ and $K$ within the KPR framework may result in improved estimates of association
between $y=\mbox{pH}$ and the taxon abundances.

For the simulation, we define a ``true" association between pH and the genus-level taxa
in $X$ using the 2-dimensional PCR model in eq.~\eqref{eq:gammaPCR} and
\eqref{eq:betaPCR}. Specifically, we use the apparent association between $y=\mbox{pH}$
and genus-level abundances in Figure~\ref{fig:EPCA}(c) to construct a ``true" coefficient
vector $\beta_{\mbox{\tiny True}}$ as follows.  Using the SVD of $X$, $X=USV'$, set
\begin{equation*}
y_{\mbox{\tiny True}}=U_{(2)}S_{(2)}\left(U'_{(2)}S_{(2)}^2U_{(2)}\right)^{-1}S_{(2)}U'_{(2)}y;
\end{equation*}
then project $y_{\mbox{\tiny True}}$ onto the space spanned by the first two singular vectors
\begin{equation*}
\beta_{\mbox{\tiny True}}=V_{(2)}S_{(2)}^{-1}U_{(2)}'y_{\mbox{\tiny True}}.
\end{equation*}

Taking $H=EE'$ in a KPR model of the form \eqref{eq:Franklin}, we compare the resulting
estimate of $\beta$ with ridge and lasso estimates. (Note that $y_{\mbox{\tiny True}}$
and $\beta_{\mbox{\tiny True}}$ are not informed by $E$.)  The simulation is repeated 500
times, each with a different $\epsilon\sim N_n(0, \sigma^2_\epsilon I_n)$ to produce
various values of $R^2\in\{0.1,0.2,...,0.9\}$. The performance metrics are the estimation
sum squared error (ESSE) the $H$-weighted prediction sum squared error (HPSSE) as in the
previous section. In this numerical example, we do not assume we always observe the true
$H$ matrix; rather, we use a noisy version, $H_{obs}$, of $H$ in KPR with
$\|H-H_{obs}\|_F/\|H\|_F\in\{0, 0.25, 0.5\}$.
For all three methods, tuning parameter values are chosen to minimize the sum of squared
test error weighted by $H_{obs}$. As in the simulation for DPCoA, we also allow for using
the largest tuning parameters such that the squared test error weighted by $H$ is within
one standard error of the minimum squared test error.

Figure~\ref{fig:EPCASim} shows that KPR significantly outperforms ridge and lasso in both
prediction and estimation. Note that even though $H$ is not used to simulate the true
association, the use of edge kernel in KPR enhances the performance of both estimation
and prediction, as long as $H$ is not severely misspecified. Once again, the performance
of \emph{estimates} from ridge and lasso improve when using a larger tuning parameter (CV
1se).

\begin{figure}
\centering
{\includegraphics[height = 6 cm]{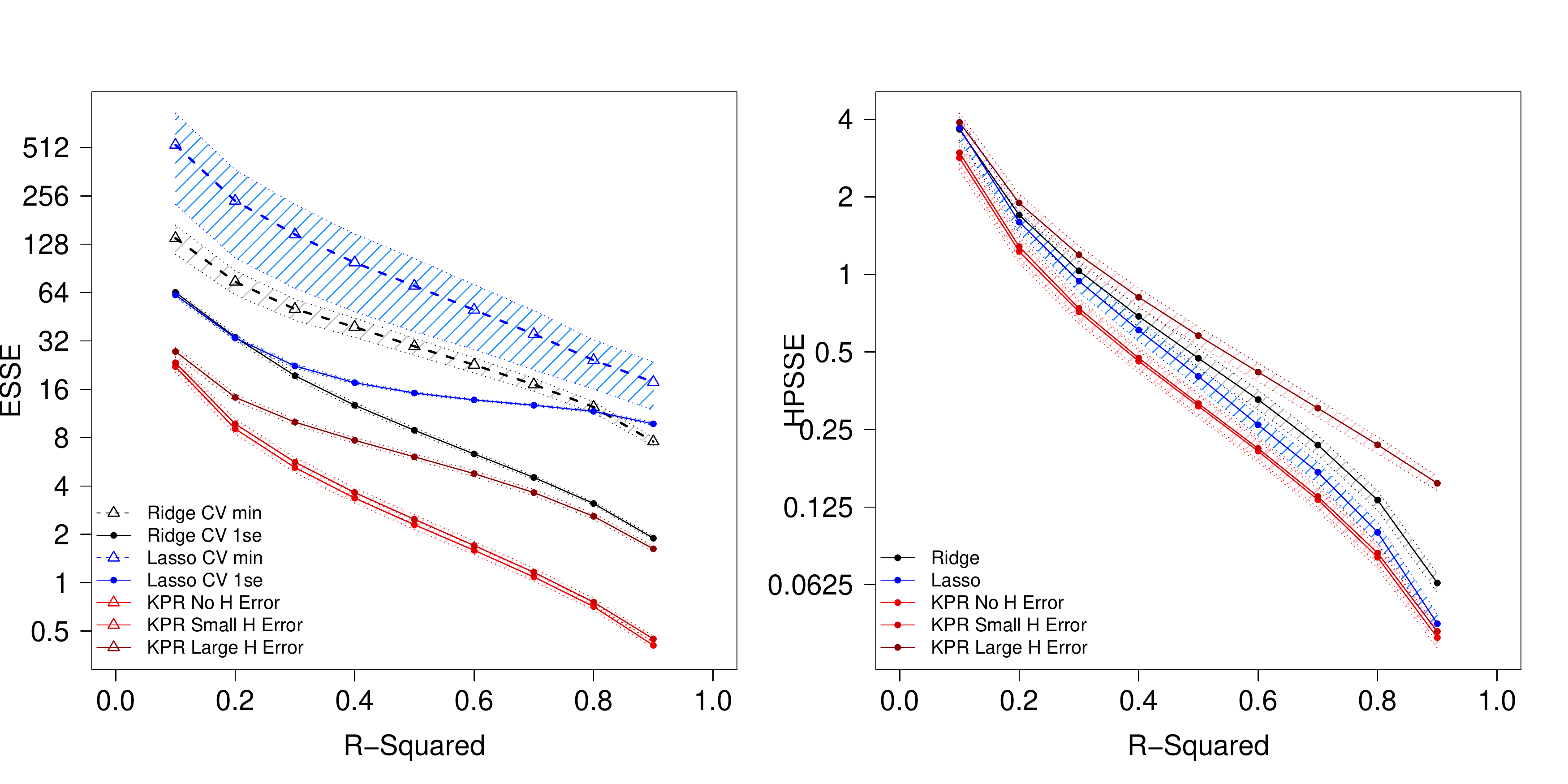}}
\caption{\small In silico evaluation of using tree-based edge information in regression models. Estimation sum squared error (ESSE) and $H$-weighted prediction sum squared error (HPSSE) of KPR (red), ridge
regression (black) and lasso (blue), with the 95\% confidence bands. For KPR, we consider
$\|H-H_{obs}\|_F/\|H\|_F = 0$ (no $H$ error), $0.25$ (small $H$ error) and $0.5$ (large
$H$ error).
}
\label{fig:EPCASim}
\end{figure}

\section{Application to an observational study}\label{sec:Applic}
We apply our kernel-penalized regression framework to  data from 16S rRNA gene collected
in a study of premenopausal women \citep{Hullar:2015}.  This study investigated aspects
of gut microbial communities in stool samples from premenopausal women
using 454 pyrosequencing of the 16S rRNA gene. The abundances of 127 species were zero
for more than 90\% of the subjects and were removed from our analysis. The data set we
consider consists of $p=128$ species sampled from $n=102$ women.

To make the measurements comparable between subjects, the species abundances were scaled
by the total number of sequences measured in each sample. This scaling produces
compositional data (the relative abundances in each sample sum to 1) which introduces
analytical complications.  In particular, regression analysis using compositional
covariates must somehow account for their unit sum constraint \citep{Kurtz:2015,Li:2014}.
For this reason, we apply the CLR transformation to the relative abundance values and use
this transformed data $\tilde{X}$ as the matrix of predictors in the KPR model.
Additionally, using Aitchison's variation matrix \citep{Aitchison:1982}, $T$, we obtain
the covariance matrix, $C$, as described prior to eq.~\eqref{eq:hatBetaC}. As $C$
provides more accurate information on the covariance among the true abundances than does
the empirical covariance matrix from relative abundances, $X$, or their CLR transform,
$\tilde{X}$, we use $C$ in place of $Q$ in \eqref{eq:hatBetaQ}.

In this example, we examine the effect of using the CLR transformed data $\tilde{X}$ and
covariance $C$ as in \eqref{eq:hatBetaC} and fit penalized regression models with the
goal of estimating $\tilde\beta_C$ in \eqref{eq:hatBetaC} for the purpose of identifying
specific species that may be associated with percent fat in the cohort described above.
To this end, we apply a recently developed significance testing procedure to three
high-dimensional models in order to identify species exhibiting evidence of association
with subjects' adiposity.  This significance test for graph-constrained estimation,
called Grace \citep{ZhaoShojaie:2015}, provides a means to assign significance to
estimates from penalized regression models that incorporate structure of the type
provided by $Q$ in \eqref{eq:hatBetaQ} (or $C$ in \eqref{eq:hatBetaC}).  The method
asymptotically controls the type-I error rate regardless of the choice of $Q$. The
special case with $Q=I$ provides a significance test for ordinary ridge regression. In
each application of the Grace test, tuning parameters are selected based on the smallest
squared test error using 10-fold cross validation. Following \citet{ZhaoShojaie:2015}, the
assumed sparsity parameter is set to be $\xi=0.05$.
The tuning parameter for the initial estimator is set to be $\lambda_{init}=4\hat\sigma_\epsilon\sqrt{3\log p/n}$, where
$\hat\sigma_\epsilon$ is the estimated standard deviation of the random error $\epsilon$,
using the scaled Lasso \citep{SunZhang:2012}.  To assess significance for the sparse
models using lasso, we apply the recently proposed significance test for lasso
regressions based on low-dimensional projection estimator (LDPE) \citep{ZhangZhang:2014,
vandegeer:2014}, which provides an asymptotically valid test for lasso-penalized
regression estimates.

\begin{table}[t]
\centering
\caption{\footnotesize
Species found  to be associated with percent fat (in increasing order of p-values) at
different significant levels using: KPR with centered log-ratio transformed abundances (CLR) ;
ridge and lasso regression with centered log-ratio transformed abundances; and ridge and
lasso regression with untransformed relative abundances (rel\%).}
\label{tab:RD}
\begin{tabular}{p{2.0cm} p{4.1cm}p{2.9cm}p{1.7cm}}
\hline \hline
	& {\footnotesize $p<0.01$}	& {\footnotesize $p<0.005$}	& {\footnotesize FDR $<0.1$}\\	\hline
{\footnotesize KPR $+$ CLR}	& {\scriptsize Bacteroides, Anaerovorax, \newline Acidaminococcus, Blautia, \newline Dethiosulfatibacter,\newline
                        Asaccharobacter, Turicibacter, \newline Lebetimonas, Streptobacillus, \newline Anoxynatronum}
                            & {\scriptsize Bacteroides, \newline Turicibacter, \newline Acidaminococcus,\newline Dethiosulfatibacter}
                            & {\scriptsize (none)}\\ \hline
{\footnotesize Ridge $+$ CLR}  & {\scriptsize (none)}	& {\scriptsize (none)}	& {\scriptsize (none)} \\ \hline
{\footnotesize Ridge $+$ rel\%} & {\scriptsize Catonella, Dethiosulfatibacter} 	& {\scriptsize (none)}	& {\scriptsize (none)} \\ \hline
{\footnotesize Lasso $+$ CLR}  & {\scriptsize Roseburia}	& {\scriptsize (none)}	& {\scriptsize (none)} \\ \hline
{\footnotesize Lasso $+$ rel\%} & {\scriptsize Dethiosulfatibacter, \newline Micropruina} 	& {\scriptsize Dethiosulfatibacter}	& {\scriptsize (none)} \\ \hline
\hline
\end{tabular}
\end{table}

We report on five regression estimation methods for which the significance of regression
coefficients can be evaluated using existing high-dimensional testing methods.  Two are
obtained using the relative abundances, $X$, with respect to: (i) an ordinary ridge
penalty and (ii) a lasso penalty. Three are obtained using the CLR transformed
abundances, $\tilde{X}$, with respect to: (iii) an ordinary ridge penalty,  (iv) a lasso
penalty,  and (v) the KPR estimate in \eqref{eq:hatBetaC}. None of these methods result
in any species associated with the outcome of percent fat when controlled for false
discovery rate (FDR) at 0.1 using the Benjamini-Yekutieli procedure \citep{BY:2001}.
However, when using a cut-off of $p=0.01$, the KPR estimate \eqref{eq:hatBetaC} results
in ten species. With a cut-off of $p=0.005$, KPR results in four species. Ordinary ridge
regressions using the CLR-transformed vectors find no associations at a cut-off of
$p=0.01$, whereas using the relative abundances, ridge finds two species at the $p=0.01$
cut-off and none at $p=0.005$. Lasso regression with the CLR-transformed vectors
identifies one specie at the $p=0.01$ cut-off and none at $p=0.005$ cut-off. When using
the relative abundances, lasso identifies two species as significant at the $p=0.01$
cut-off and one at the $p=0.005$ cutoff. See Table~\ref{tab:RD} for the list of
identified species.

\section{Discussion}\label{sec:Discussion}
We have formulated a family of regression models that naturally extends the
dimension-reduced graphical explorations common to microbiome studies. In this sense, we
have simply re-focused the role of the eigen-structures used in ordination methods toward
exploiting this structure in penalized regression models. The large family of models
developed here provides a supervised statistical learning counterpart to the unsupervised
methods of principal coordinate analysis (PCoA).

A primary motivations for PCoA graphical displays is the ability to incorporate
biologically-inclined measures of (dis)similarity. The popular use of UniFrac, for
instance, is motivated by the desire to impose phylogeny into the analysis. These
dissimilarities have also been used for rigorous statistical testing in the context of
Anderson's nonparametric MANOVA \citep{Anderson:2006} or the closely-related kernel
machine regression score test \citep{Chen:2012,Pan:2011,NiZhao:2015} for global
association of a multivariate predictor with an outcome. However, the use of UniFrac and
other non-Euclidean distances make it difficult to identify specific associations between
the microbial abundance profiles and a phenotype; indeed, none of these analyses proceed
to estimate the  individual associations. In addition to ordination displays and global
tests for associations, a variety of machine learning approaches have emphasized on
models that \emph{predict} a response. In contrast, we focus on \emph{estimating} the
coefficient vector, which is a key aspect of any approach used to draw scientific
conclusions based on the association of microbial communities with an outcome or
phenotype.

An interesting feature of the proposed kernel-penalized regression framework is its
ability to sidestep some of the problems inherent in compositional data analysis.
Indeed, as emphasized by \citet{Li:2014} regression analysis with compositional
covariates must somehow acknowledge their unit-sum constraint and spurious correlations.
Our approach, which differs somewhat from that of \citet{Li:2014}, may also be viewed as
a penalized version of the low-dimensional linear model for compositions by
\citet{Tolosana-Delgado-Chap26:2011}, who use the isometric log-ratio (ILR) coordinates.
We note that ILR coordinates arise from the SVD of mean-centered CLR-transformed data,
$\tilde{X}$ (see \citet{Egozcue-Chap2:2011}), which is also used in our model. However,
to estimate $\beta\in\bbR^p$, we used instead a regularization framework; our penalty in
Section~\ref{sec:compdata} arises from Aithison's total variation matrix whose singular
values are the total variances of ILR components. Moreover, the proposed framework also
allows us to use existing inference frameworks for high-dimensional regression, and in
particular the Grace test \citep{ZhaoShojaie:2015}, to  assess the significance of
estimated regression coefficients.

\newpage
\bibliography{KernelRegn}
\bibliographystyle{imsart-nameyear}

\end{document}